\documentclass[times,twocolumn,final]{elsarticle}

\usepackage{medima}

\usepackage{outline}
\usepackage{amssymb}
\usepackage{amsmath}
\usepackage{graphicx}
\usepackage{times}
\usepackage{bm} 
\usepackage{url}

\usepackage{epstopdf}

\usepackage{epsfig}
\usepackage{tikz}
\usetikzlibrary{spy}
\usepackage{algpseudocode}
\usepackage{algorithm}
\usepackage{mathrsfs}

\usepackage{nicefrac}       
\usepackage{booktabs}       

\usepackage{amsmath,graphicx,caption,mathtools,amssymb,amsthm}
\usepackage{graphicx}
\usepackage{multirow}
\usepackage{subcaption}
\usepackage{accents}

\newtheorem{proposition}{\bf{Proposition}}

\long\def\comment#1{} 

\newcommand{\xmath}[1] {\ensuremath{#1}\xspace}
\newcommand{\blmath}[1] {\xmath{\bm{#1}}}

\newcommand{\x}{\blmath{x}}

\newcommand{\xb}{{\blmath x}}
\newcommand{\yb}{{\blmath y}}
\newcommand{\zb}{{\blmath z}}

\newcommand{\Ac}{\mathcal{A}}

\newcommand{\Tc}{\mathcal{T}}
\newcommand{\Xc}{\mathcal{X}}
\newcommand{\Yc}{\mathcal{Y}}

\newcommand{\Rd}{{\mathbb R}}

\newcommand{\Pc}{{{\mathcal P}}}

\newcommand{\Kd}{\mathbb{K}}

\newcommand{\beq}{\begin{equation}}
\newcommand{\eeq}{\end{equation}}
\newcommand{\beqa}{\begin{eqnarray}}
\newcommand{\eeqa}{\end{eqnarray}}

\usepackage{amsmath}
\usepackage{framed,multirow}
\usepackage{amssymb}
\usepackage{latexsym}
\usepackage{url}
\usepackage{xcolor}
\usepackage{xspace}

\usepackage{hyperref}

\definecolor{newcolor}{rgb}{.8,.349,.1}

\journal{Medical Image Analysis}

\begin{document}

\verso{Hyungjin Chung \textit{et~al.}}

\begin{frontmatter}

\title{Two-Stage Deep Learning for Accelerated 3D Time-of-Flight MRA without Matched Training Data}

\author[1]{Hyungjin Chung}
\author[1]{Eunju Cha}
\author[2]{Leonard Sunwoo\corref{cor1}}
\ead{leonard.sunwoo@gmail.com}
\author[1]{Jong Chul Ye\corref{cor1}}
\cortext[cor1]{Corresponding authors.}
\ead{jong.ye@kaist.ac.kr}

\address[1]{Department of Bio and Brain Engineering, Korea Advanced Institute of Science and Technology (KAIST), Daejeon 34141, Republic of Korea}
\address[2]{Department of Radiology, Seoul National University College of Medicine, Seoul National University Bundang Hospital, Seongnam, Republic of Korea}

\received{?}
\finalform{?}
\accepted{?}
\availableonline{?}
\communicated{?}

\begin{abstract}
Time-of-flight magnetic resonance angiography (TOF-MRA) is one of the most widely used non-contrast MR imaging methods to visualize blood vessels, 
but due to the 3-D volume acquisition highly accelerated
acquisition is necessary.
Accordingly,
high quality reconstruction from undersampled TOF-MRA is an important research topic for deep learning.
However, most existing deep learning
works require matched reference data for supervised training, which are often difficult to obtain.
By extending the recent theoretical understanding of cycleGAN from the optimal transport theory,
here we propose a novel two-stage  {unsupervised} deep learning approach, 
  which is composed of the  multi-coil reconstruction network along the coronal plane followed by
a multi-planar refinement network along the axial plane.
Specifically, the first network is trained in the square-root of sum of squares (SSoS) domain to achieve high quality parallel image reconstruction, 
whereas  the second refinement network is designed 
to efficiently learn the characteristics of highly-activated blood flow using double-headed max-pool discriminator. 
Extensive experiments 
demonstrate that the proposed  learning process without matched reference exceeds performance of state-of-the-art compressed sensing (CS)-based method and provides comparable or even better results than supervised learning approaches. 
\end{abstract}


\end{frontmatter}


\section{Introduction}
\label{sec: intro}
Time-of-flight magnetic resonance angiography (TOF MRA)\citep{keller1989mr, miyazaki2012non, wheaton2012non, laub1995time} is widely used in clinical situations for  visualizing blood flow without the need for the injection of contrast agents.
Here,  the phenomenon of flow-related enhancement of spins entering into an imaging slice is exploited to amplify the contrast between blood vessels and surrounding tissues.
 
In 2-D TOF, multiple thin imaging slices are acquired with a flow-compensated gradient-echo sequence, whereas in 3D TOF a volume of images is obtained simultaneously by phase-encoding in the slice-select direction. These images can be then combined using the maximum intensity projection (MIP) so that one can obtain a 3-D image of the vessels analogous to conventional angiography.  Accordingly,
TOF MRA provides tremendously helpful physiological information for the detection of stenosis or occlusion in the intracranial arteries.

When taking scans of TOF MRA, fully acquiring $k$-space is painfully time consuming, especially for 3-D scans where a large volume has to be covered.
Furthermore, patient motion during the stretched scan time causes artifacts in the image. Consequently, accelerating MR scans would lead to increase the patient throughput and relieve the issue of motion artifacts.

\subsection{CS-MRI and pMRI}

To reduced the long scan time, $k$-space can be sub-sampled, but the $k$-space under-sampling  subsequently introduces aliasing artifacts. To resolve this issue, multiple receiver coils can be utilized to merge information from different receiver coils to compensate
for the missing $k$-space data. These parallel MRI (pMRI)  \citep{SENSE, GRAPPA} techniques are routinely used  in clinical practice. 

For the TOF MRA,  compressed sensing (CS) algorithms \citep{lustig2007sparse,jung2009k} have been  also extensively studied by exploiting the sparsity in the original image domain, which is an inherent nature of angiograms. Moreover, 
applications of CS in conjunction with pMRI have been extensively investigated \citep{Stalder2015CSMRA, Hutter2015CSMRA, Tang2019CSMRA,jin2016general}. Although CS-MRI have shown its effectiveness in the reconstruction of MRI, the inherently iterative nature of the method leads to slow and expensive computation. Moreover, its inability to \textit{learn} from given data distribution is also a drawback.

\subsection{Deep Learning for CS-MRI}

Recently, a myriad of deep learning algorithms have been proposed  for MR reconstruction, which show superior performance over CS-MRI while significantly reducing computation time \citep{wang2016accelerating, schlemper2017deep, zhu2018image, schlemper2018stochastic, eo2018kiki, wang2019dimension, liu2019ifr, wang2020deepcomplexmri, sriram2020end,  lee2018deep, han2017deep, hammernik2018learning}. Generative adversarial networks (GAN)\citep{goodfellow2014generative} have also been largely investigated in the context of MR reconstruction \citep{mardani2017deep, wang2019accelerated, quan2018compressed, yang2017dagan} to further enhance the reconstruction quality. 

Nonetheless, most of the deep learning approaches are supervised learning framework where a large amount of matched fully sampled scans must be provided to train the neural network properly. This imposes fundamental challenges in neural network trainings, since the matched fully sampled
 reference data should be
acquired under the same conditions, which is not always possible in clinical environment.

\subsection{Our contributions}

In our recent paper \citep{sim2019OT}, 
we proposed a systematic  framework to design various types of unsupervised learning architecture for general inverse problems using
the optimal transport theory  \citep{villani2008optimal,peyre2019computational},
and also provided preliminary results for single coil 2D MR reconstruction from sparse Fourier samples \citep{sim2019OT}.
The resulting network architecture is similar to cycleGAN \citep{zhu2017unpaired},
but the knowledge of the imaging physics can significantly simplify the network architecture and training scheme \citep{sim2019OT}.

By extending this idea, 
here we suggest a novel unpaired  multiplanar deep learning scheme  which aims specifically at the reconstruction of under-sampled 3D TOF MRA scan.
To overcome the large GPU memory and training data requirement for 3-D learning,
we propose a novel  architecture that  consists of two successive unsupervised training steps in 2D space. The first step is the reconstruction of MRA scan in the coronal plane, which is done slice by slice, incorporating complex multi-coil data into the training scheme. In the second step of reconstruction, we aim to further enhance the quality of reconstruction, especially in terms of maximum intensity projection (MIP) images, through the use of stacked 3D reconstruction with the newly introduced \textit{projection discriminator}.
One of the important advantages of the proposed two-stage unsupervised learning scheme is that
each neural network can be trained with different sets of unpaired training data set, which maximizes the utility of
available data for training purpose.

In brief, our contributions can be summarized as follows:

\begin{itemize}
	\item Two-stage unsupervised learning process for 3D reconstruction, in the coronal plane and the axial plane respectively, is proposed. This sequential learning process is particularly useful in 3D MR acceleration where you have 4 dimensions (3 spatial, 1 for coil). 
	\item {Projection discriminator}, which learns the distribution of both volumetric and max-pooled images, is proposed. The discriminator is used in the second stage of reconstruction, and proves to enhance the quality of images greatly, especially in terms of MIP images.
	\item By deriving network architectures using the optimal transport theory, unwanted artificial features, which are often observed in GAN type algorithms, can be prevented in a top-down manner.
\end{itemize}

The remainder of the paper is organized as follows: in Section~\ref{sec:review}, we briefly review the geometry of cycleGAN from optimal transport
theory perspective; in Section~\ref{sec:theory}, the theory of our two-step unsupervised learning framework for 3D TOF MRA is proposed by 
extending the theory of OT driven cycleGAN. In Section~\ref{sec:method}, exhaustive description of methods and materials is provided. In Section~\ref{sec:results}, experimental results in both in-vitro and in-vivo situations are shown. In Section~\ref{sec:discussion}, we discuss different choices for the design of our learning process,
which is followed by conclusions in Section~\ref{sec:conclusion}.

\section{Related Works}
\label{sec:review}

In this section, to make the paper self-contained, we will briefly review the optimal transport driven cycleGAN proposed in our companion paper \citep{sim2019OT}.

\subsection{Geometry of CycleGAN}

%
%
%
 
 Consider the following measurement model:
\begin{eqnarray}
\yb&=&F \xb \ , 
\end{eqnarray}
where $\yb \in \Yc$ and $\xb \in \Xc$ denote the measurement and the unknown image, respectively, and $F : \Xc \mapsto \Yc$ is the imaging operator,
which could be known, partially known, or completely unknown. 

In contrast to the supervised learning where the goal  is to learn the relationship between  the  image $\xb$ and measurement $\yb$ pairs,
in the unsupervised learning framework there are no matched image-measurement pairs. 
Still we could have sets of images and unpaired measurements, so the goal of unsupervised learning is to match the probability distributions rather than each individual samples
as shown in Fig.~\ref{fig:cycleGANgeom}. 
This can be done by finding  transportation maps that transport the probability measures between the two spaces.

\begin{figure}[!hbt] 	
\center{ 
\includegraphics[width=8cm]{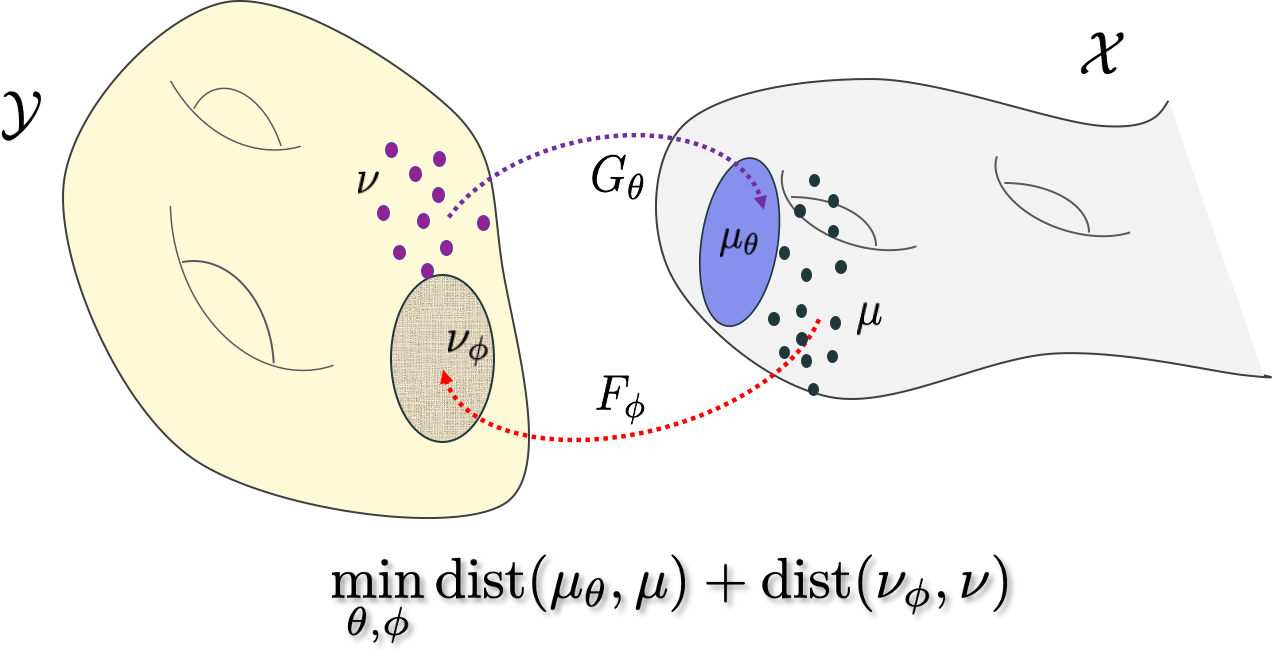}
}
\caption{Geometric view of unsupervised learning.}
\label{fig:cycleGANgeom}
\end{figure}

Specifically, suppose that the target image space
$\Xc$ is equipped with a probability measure $\mu$, whereas
the measurement space  $\Yc$ is  with a probability measure $\nu$ as shown in Fig.~\ref{fig:cycleGANgeom}.
Then, 
we can see that the  mass transport from $(\Xc,\mu)$ to $(\Yc,\nu)$ is performed by the forward operator $F$,
 so that  $F$ ``pushes forward'' the measure $\mu$ in $\Xc$ to $\nu_F$ in the space $\Yc$ \citep{villani2008optimal,peyre2019computational}.
On the other hand, the mass transportation from the measure space $(\Yc,\nu)$ to another measure space $(\Xc,\mu)$ is done by a generator $G: \Yc \mapsto \Xc$, i.e.
the generator $G$ pushes forward the measure $\nu$ in $\Yc$ to a measure $\mu_G$ in the target space $\Xc$. 
Then, the optimal transport map for unsupervised learning can be achieved by minimizing the statistical distances  $\mathrm{dist}(\mu,\mu_G)$ between $\mu$ and $\mu_G$, and  $\mathrm{dist}(\nu,\nu_F)$ 
 between $\nu$ and $\nu_F$, and our proposal is to use the Wasserstein-1 metric as a means to measure the statistical distance.

More specifically, for the choice of a metric $d(\xb,\xb')=\|\xb-\xb'\|$ in $\Xc$,  
the Wasserstein-1 metric between $\mu$ and $\mu_G$ can
be computed by \citep{villani2008optimal,peyre2019computational}
\begin{align}\label{eq:Wmu}
W_1(\mu,\mu_G)
=&\inf\limits_{\pi \in \Pi(\mu,\nu)}\int_{\Xc\times \Yc} \|\xb-G(\yb)\|d\pi(\xb,\yb) 
\end{align}
where $\Pi(\mu,\nu)$ is the set of joint measures whose marginal distributions in $\Xc$ and $\Yc$ are $\mu$ and $\nu$, respectively.
Similarly, the Wasserstein-1 distance between $\nu$ and $\nu_F$ is given by
\begin{align}\label{eq:Wnu}
W_1(\nu,\nu_F)
=&\inf\limits_{\pi \in \Pi(\mu,\nu)}\int_{\Xc\times \Yc} \|F(\xb)-\yb\|d\pi(\xb,\yb) 
\end{align}
Since our goal is to find the transportation maps represented by the joint distribution $\pi$, 
separate minimization of  \eqref{eq:Wmu} and \eqref{eq:Wnu}  is not desirable; instead, 
we should  minimize them together with the same joint distribution $\pi$:
\begin{align}\label{eq:unsupervised}
\inf\limits_{\pi \in \Pi(\mu,\nu)}\int_{\Xc\times \Yc}c(\xb,\yb;G,F) d\pi(\xb,\yb) 
\end{align}
where the transportation cost is defined by
\begin{align}\label{eq:ourc}
c(\xb,\yb;G,F)= \|\xb-G(\yb)\|+ \|F(\xb)-\yb\|
\end{align}

One of the most important contributions of our companion paper \citep{sim2019OT} is to show that
the primal formulation of the unsupervised learning in \eqref{eq:unsupervised}  with the transport cost \eqref{eq:ourc}
can be represented by a dual formulation:
\begin{eqnarray}\label{eq:OTcycleGAN}
\min_{G,F}\max_{\psi,\varphi}\ell_{cycleGAN}(G,F;\psi,\varphi)
\end{eqnarray}
where 
\begin{eqnarray}
\ell_{cycleGAN}(G,F;\psi,\varphi):=  \lambda \ell_{cycle}(G,F) +\ell_{Disc}(G,F;\psi,\varphi) 
\end{eqnarray}
where $\lambda>0$ is the hyper-parameter, and  the cycle-consistency term is given by
\begin{align}\label{eq:cycleloss} 
\ell_{cycle}(G,F)  =& \int_{\Xc} \|\xb- G(F(\xb)) \|  d\mu(\xb) \\
&+\int_{\Yc} \|\yb-F(G(\yb))\|   d\nu(\yb) \notag 
\end{align}
whereas  the second term is the discriminator term:
\begin{align}
&\ell_{Disc}(G,F;\psi,\varphi)  \label{eq:Disc} \\
=&\max_{\varphi}\int_\Xc \varphi(\xb)  d\mu(\xb) - \int_\Yc \varphi(G(\yb))d\nu(\yb) \notag \\
 & + \max_{\psi}\int_{\Yc} \psi(\yb)  d\nu(\yb) - \int_\Xc \psi(F(\xb))  d\mu(\xb) \notag
\end{align}
Here, $\varphi,\psi$ are often called Kantorovich potentials and satisfy 1-Lipschitz condition (i.e.
\begin{align*}
|\varphi(\xb)-\varphi(\xb')|\leq \|\xb-\xb'\|,&~\forall \xb,\xb'\in \Xc \\
|\psi(\yb)-\psi(\yb')|\leq \|\yb-\yb'\|,&~\forall \yb,\yb'\in \Yc
\end{align*}
We further showed that if the forward operator $F$ is known, the optimization
with respect to $F$ in \eqref{eq:OTcycleGAN} is no more necessary, which leads to the
simplified discriminator term:
\begin{align}
\ell_{Disc}(G,F;\varphi)  
=\max_{\varphi}\int_\Xc \varphi(\xb)  d\mu(\xb) - \int_\Yc \varphi(G(\yb))d\nu(\yb) \label{eq:discsimple}  
\end{align}
We will show that these two forms of optimal transport driven cycleGAN (OT-cycleGAN) is useful for the proposed
two-stage reconstruction method.

\section{Theory}
\label{sec:theory}

\subsection{Forward Model}
 
%
%

One of the most widely used 3D TOF techniques is the so called
MOTSA, which stands for Multiple Overlapping Thin Slab Acquisition \citep{blatter1991cerebral}.  
MOTSA involves the sequential acquisition of a several overlapping 3D volumes (or ``slabs"). Each slab contains
relatively small number of slices, so loss of signal due to saturation effects is relatively limited.
However,  some variation in signal still occurs at the end slices due to the saturation effect, so MOTSA extracts only the central portions for each of the overlapping acquisitions to make up the final data set for processing into the MRA projections. 
The end slices are typically discarded or averaged with those in the adjacent MOTSA section. 

\begin{figure}[!hbt]
\hspace{-1.0cm}
\center{ 
\includegraphics[width=8.0cm]{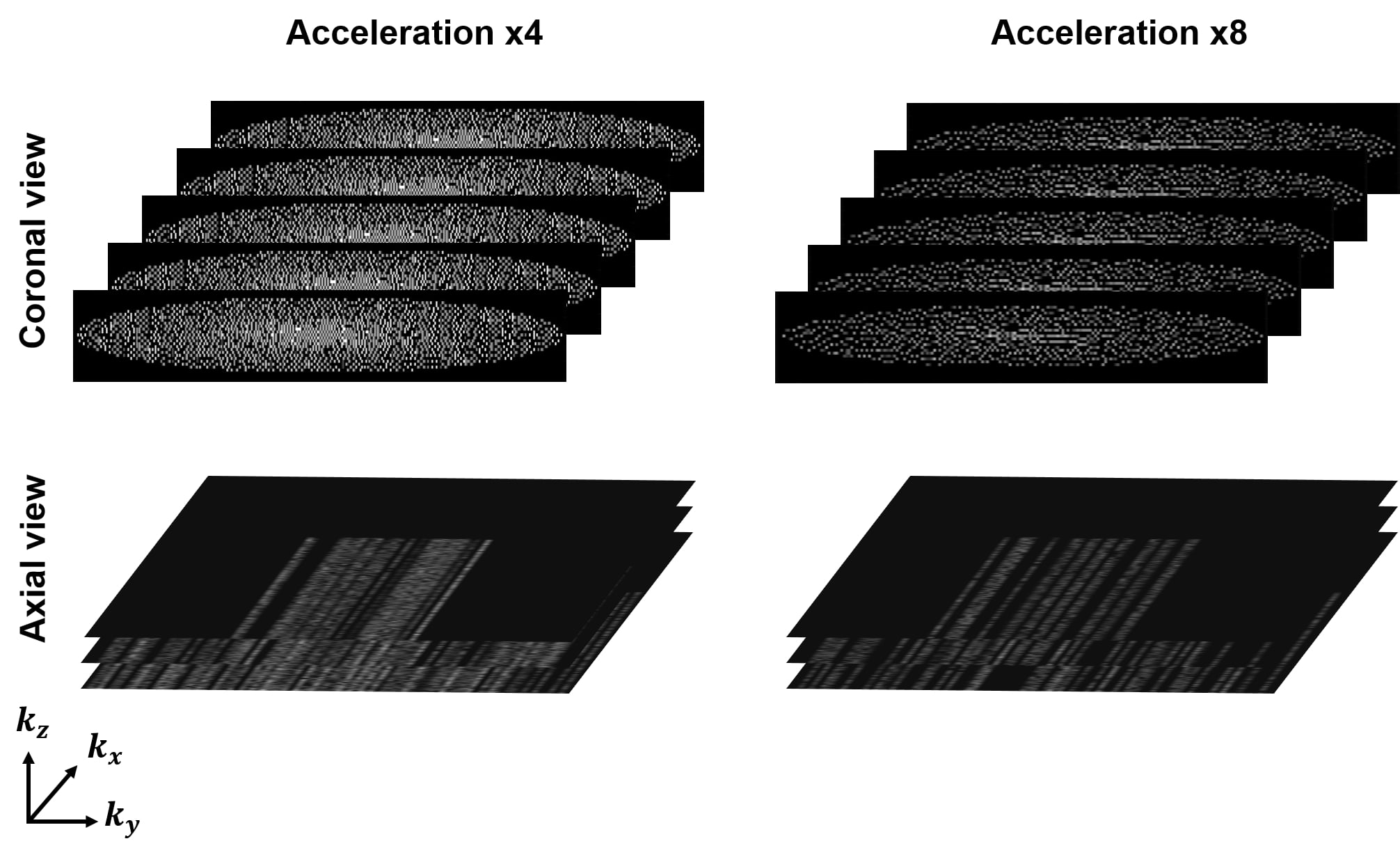}
}
\caption{Sampling masks used for both prospective and retrospective under-sampling, each responsible for $\times$4 and $\times$8 acceleration. The first row visualizes masks in the coronal plane. The second row shows masks in the axial plane, where partial Fourier sampling scheme \citep{feinberg1986halving} was applied.}
\label{fig:sampling_mask}
\end{figure}

In accelerated MOTSA acquisition,
3D scans, when seen from the coronal plane, have the same sampling mask
 specifically given in Fig.~\ref{fig:sampling_mask}. 
Performing Fourier transform along the read-out direction leads to the following forward problem:
 \begin{eqnarray}\label{eq:fwd}
\widehat \xb & =\Pc_\Omega\Tc \xb
\end{eqnarray}
where with a slight abuse of notation we define
\begin{eqnarray}\label{eq:fwdspec}
\begin{split}
\xb: = \begin{bmatrix}
\xb^{(1)} & \cdots & \xb^{(C)}
\end{bmatrix}, \quad
\widehat 
\xb:=\begin{bmatrix} \widehat\xb^{(1)} & \cdots & \widehat\xb^{(C)}\end{bmatrix}
\end{split}
\end{eqnarray}
in which  $C$ is the number of coils,
$\Tc$ denotes 2D spatial Fourier transform,  and $\Pc_\Omega$ is the projection operator on the sampling
mask $\Omega$ such as Fig.~\ref{fig:sampling_mask}.

\begin{figure*} 	
\center{ 
\includegraphics[width=18.0cm]{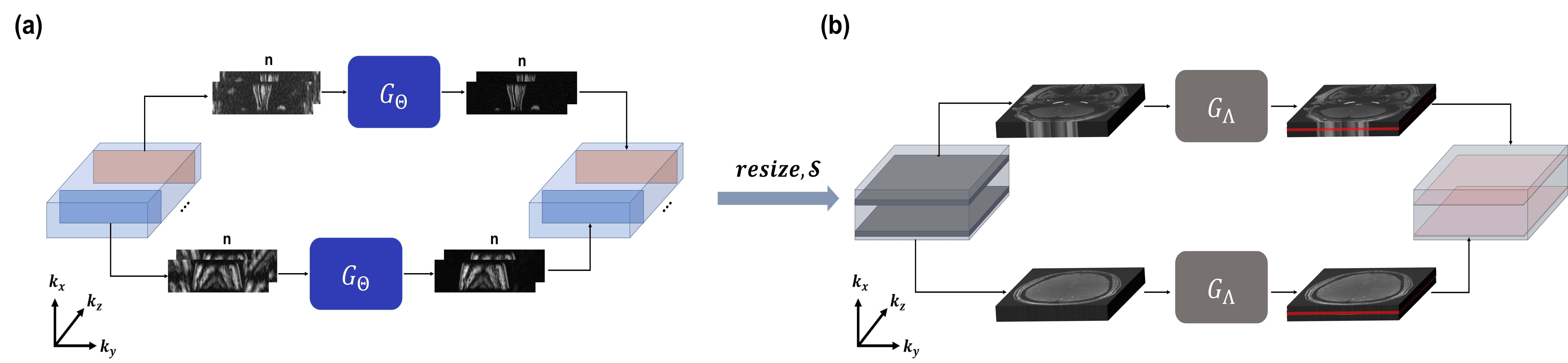}
}
\caption{Overall pipeline of multi-planar learning scheme. (a) Step I: Coronal reconstruction - images are reconstructed slice-by-slice with $G_\Theta$ which are then stacked to form a full volume. Resizing from matrix size 774$\times$359$\times$21 to 512$\times$512$\times$45 is done, and the coil dimension is merged from SSOS ($\Ac$) operation. (b) Step II: Axial reconstruction - volume data of slice depth 7 are fed to $G_\Lambda$, while only the center slices from reconstruction output are used to refine each slice of the volume.}
\label{fig:flowchart}
\end{figure*}

\begin{figure*}[ht!]
\center{ 
\includegraphics[width=18.0cm]{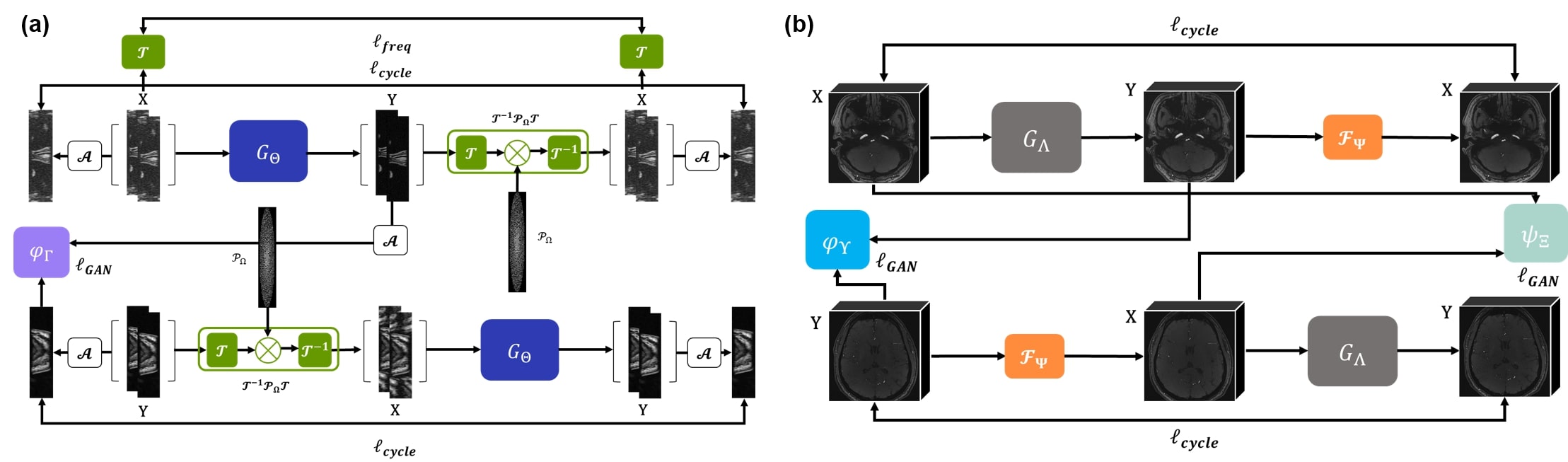}
}
\caption{Detailed pipeline of each training scheme. (a) Step I: coronal reconstruction. Complex valued multi-coil data are trained with MR-physics driven cycleGAN. All the losses in the image domain, i.e. $\ell_{cycle}, \ell_{GAN}$ are calculated with respect to SSOS images. (b) Step II: axial reconstruction. multi-coil information is merged prior to the second step of reconstruction process. Training in axial direction is done partially in 3D, which consists of 7 slices of stacked images.}
\label{fig:detailed_flowchart}
\end{figure*}

\subsection{Two Step Unsupervised 3D TOF Reconstruction}

For a given forward model in \eqref{eq:fwd}, which is obtained from sampling scheme along the coronal plane as in 
Fig.~\ref{fig:sampling_mask},
the reconstruction should be also performed in the coronal direction.
Unfortunately, this poses a problem since the radiologists would typically review images in the axial plane and the reconstruction plane is not aligned with the viewing plane by radiologists; thus,
remaining reconstruction artifacts from the coronal direction may reduce diagnostic performance.
One could address this using 3D learning, but the memory requirement for 3D neural network training
is much larger than the standard GPU memory, which prohibits its use.

Therefore, the main idea of the proposed method is a two step approach,
where  the first step reconstruction is performed along the coronal direction, which is followed by the second
step for the
axial directional refinement, as shown in Fig.~\ref{fig:flowchart}.
In particular, our emphasis is to perform these two step reconstruction without matched reference data,
where the following Proposition is useful in network design.
%

\begin{proposition}\label{prp:main}
Suppose that the transportation cost for the primal OT problem in \eqref{eq:unsupervised}
is given by
\begin{align}
c(\xb,\yb;G,F)=&   \|\Ac(\xb)-\Ac(G(\yb))\|+ \|\Ac(F(\xb))-\Ac(\yb)\|  \notag  \\
&+ b_x(\xb;G,F)+ b_y(\yb;G,F) \label{eq:cmain}
\end{align}
where $\Ac$ is a deterministic (non)linear operator, and $b$ and $c$ are deterministic (non)linear functionals, i.e. $b_x:\Xc\mapsto \Rd$ and $b_y:\Yc\mapsto \Rd$. Then, the corresponding
dual OT problem is given by
\begin{eqnarray}\label{eq:dualOT}
\min_{G,F}\max_{\psi,\varphi}\ell_{dualOT}(G,F;\psi,\varphi)
\end{eqnarray}
where 
\begin{align}
&\ell_{dualOT}(G,F;\psi,\varphi):=  \\
& \lambda \ell_{cycle}(G,F) +\ell_{Disc}(G,F;\psi,\varphi)  + \ell_x(G,F) + \ell_y(G,F) \notag
\end{align}
where $\lambda>0$ is the hyper-parameter, and  the cycle-consistency term is given by
\begin{align*}
\ell_{cycle}(G,F)  =& \int_{\Xc} \|\Ac(\xb)- \Ac(G(F(\xb))) \|  d\mu(\xb) \\
&+\int_{\Yc} \|\Ac(\yb)-\Ac(F(G(\yb)))\|   d\nu(\yb)
\end{align*}
whereas  the second term is the discriminator term:
\begin{align}
&\ell_{Disc}(G,F;\psi,\varphi)  \label{eq:disc} \\
=&\max_{\varphi}\int_\Xc \varphi(\Ac(\xb))  d\mu(\xb) - \int_\Yc \varphi(\Ac(G(\yb)))d\nu(\yb) \notag \\
 & + \max_{\psi}\int_{\Yc} \psi(\Ac(\yb))  d\nu(\yb) - \int_\Xc \psi(\Ac(F(\xb)))  d\mu(\xb) \notag
\end{align}
with 1-Lipschitz function $\varphi,\psi$, and the last two terms are given by
\begin{align*}
 \ell_x(G,F) :=& \int b_x(\xb;G,F) d\mu(\xb)\\
 \ell_y(G,F) :=& \int b_y(\yb;G,F) d\nu(\yb) 
\end{align*}
\end{proposition}
\begin{proof}
See Appendix.
\end{proof}

\subsubsection{Step I: Coronal Reconstruction}


Using Proposition~\ref{prp:main} we are now ready to derive our algorithm.
First,  to make the dimension of $\Xc$ and $\Yc$ the same,
the forward model in \eqref{eq:fwd} is first converted
to an image domain forward formulation by taking inverse Fourier transform:
\begin{eqnarray}\label{eq:fwd_1}
\yb & = \Tc^{-1}\Pc_\Omega\Tc \xb = F\xb,\quad \mbox{with}\quad F:= \Tc^{-1}\Pc_\Omega\Tc 
\end{eqnarray}
where  $\yb = \begin{bmatrix} \yb^{(1)} & \cdots & \yb^{(C)}\end{bmatrix}$ 
and
$\Tc^{-1}$ is the inverse Fourier transform.
Then,  define the following transportation cost:
\begin{align} 
    c(\xb, \yb;G, F) &= \|\Ac(\yb)- \Ac (F(\xb))\| \label{eq:fidelity}\\
    &+ \|\Ac(\xb) - \Ac(G(\yb)) \| \label{eq:DLprior}\\
    &+ \alpha\| \Ac(\xb) - \Ac(G(\xb))\| \label{eq:ident}\\
    &+ \beta\| \Pc_\Omega \Tc \xb - \Pc_\Omega \Tc G(F(\xb))\|_F ^2\label{eq:kfidelity}
\end{align}
where  $\alpha$ and $\beta$ are appropriate hyperparameters,
and $\Ac$ is now defined as  the square-root of sum of squares (SSOS) operation $\zb = \Ac(X)$ for multi-coil data, 
where the $n$-th component of the vector $\zb$ is formally defined as:
\begin{align}\label{eq:ssos}
z_n =  \left(\sum_{i=1}^C |x^{(i)}_n|^2 \right)^{\frac{1}{2}}
\end{align}

The transportation cost $c(\xb, \yb;G, F)$ deserves further discussion.
Specifically, the first two terms \eqref{eq:fidelity} and \eqref{eq:DLprior} are directly related to those in OT-cycleGAN,
but the loss is calculated after taking the SSoS to make the image comparison less dependent on the coil
sensitivity map.
On the other hand, the identity loss
 \eqref{eq:ident} enforces regularization to the neural network such that it does not alter images that are already in the $\Yc$ domain,
 and  \eqref{eq:kfidelity} refers to data fidelity term in the  k-space domain.
To apply data consistency to k-space data that are inherently acquired in complex domain for each coil, 
we calculate the k-space loss using Frobenius norm.

By inspection, we can see that our transportation cost is identical to \eqref{eq:cmain}
if we set
\begin{align}
&b_x(\xb;G,F) \notag\\
&:=  \alpha\| \Ac(\xb) - \Ac(G(\xb))\| + \beta\| \Pc_\Omega \Tc \xb - \Pc_\Omega \Tc G(F(\xb))\|_F ^2 \\
&b_y(\xb;G,F) =0
\end{align}
so that we can use the dual formulation in Proposition~\ref{prp:main}.
Moreover, since  the k-space sampling mask $\Omega$ is known a priori,   the competition between $F$ and $\psi$ is not necessary and we only need
to estimate $G$ and the corresponding discriminator $\varphi$.
By modeling  them with neural networks with parameters $\Theta$ and $\Gamma$,  respectively,
we can obtain the following loss function:
\begin{eqnarray}
 \min_\Theta \max_\Gamma \ell(\Theta, \Gamma)
\end{eqnarray}
with
\begin{align}\label{eq:phaseIloss}
\ell(\Theta, \Gamma) &= \gamma \ell_{cycle}(\Theta) + \ell_{Disc}(\Theta, \Gamma)\\
&+ \alpha \ell_{identity}(\Theta) + \beta \ell_{freq}(\Theta).
\end{align}
where $\gamma,\alpha$ and $\beta$ denote some hyper-parameters, and
\begin{align}
\label{eq:simple_cycle_loss}
\ell_{cycle}(\Theta)  &= \int_{\Yc} \|\Ac(\yb) - \Ac(F (G_\Theta(\yb)))\| d\nu(\yb) \notag\\
& +    \int_{\Xc} \|\Ac(\xb) - \Ac (G_\Theta(F(\xb)))\| d\mu(\xb)  ,
\end{align}
and
\begin{align}
\label{eq:simple_wgan_loss}
\ell_{Disc}(\Theta,\Gamma)  &= \int_{\Xc} \varphi_{\Gamma}(\Ac(\xb))d\mu(X)- \int_{\Yc} \varphi_{\Gamma}(\Ac(G_\Theta(\yb))) d\nu(\yb)  \\
\label{eq:identity}
\ell_{identity}(\Theta)  &= \int_{\Xc} \|\Ac(\xb)- \Ac(G_\Theta(\xb)) \| d\mu(\xb)\\  
\label{eq:freq}
\ell_{freq}(\Theta)  &= \int_{\Yc} \|\Pc_\Omega\Tc \xb - \Pc_\Omega \Tc G_\Theta(F(\xb)) \|_F^2 d\mu(\xb)
\end{align}

\subsubsection{Step II: Axial Reconstruction}

After the reconstruction through Step I, outputs are stacked together to form a single slab. As will be shown later in experiments,
when we see the images in the axial plane, however, images tend to be blurry, and lacks proper texture. Accordingly, when MIP is performed, thin vessel structures are omitted, or disconnected, which may lead to misdiagnoses such as vascular stenosis. 

Consequently, we devise a method for axial image enhancement which utilizes another unsupervised neural network to improve the quality  especially in MIP images.
 More specifically, as shown in Fig.~\ref{fig:flowchart}(b), after the reconstruction in the coronal plane,
 we construct a 3D volume for each slab, which is used as input for axial image refinement network.
   The rationale for taking stacked volume as input are as follows: first, with the use of volume data, we can perform MIP to the  volume, so that the networks can learn the distribution of the partially projected image. Second, being able to infer from adjacent slices, the network can take advantage of information from bordering slices. The advantages will be discussed more thoroughly in the discussion section.

One thing to note here is that the relationship between the input and output domains in Step II is not well-defined. More specifically, with a slight
abuse of notation, let $\Yc$ be the distribution of 3-D volume of SSoS images that were reconstructed through Step I, and $\Xc$ be the desired 3-D volume of SSoS image distribution. Unlike Step I, where we could replace one of the generators with a known forward operator, there exists no closed form
mapping $F: \Xc \mapsto \Yc$ in this case due to the SSoS operation and volume stacking.
This situation corresponds to the OT-cycleGAN formulation  where both forward and 
inverse operators are unknown.
More specifically, by defining the forward operator $F$ in terms of  a  neural network parameterized by  $\Psi$, we define
the following transportation cost
\begin{eqnarray}
c(\xb,\yb;G,F) = \|\yb - F(\xb)\| + \|G(\yb) - \xb\|,
\end{eqnarray}
Then, the corresponding OT-cycleGAN formulation is given as a Kantorovich dual formulation in \eqref{eq:OTcycleGAN}
where $\ell_{cycle}$, and $\ell_{Disc}$ are cyclic consistency loss and Wasserstein GAN  loss, respectively, which are represented by \eqref{eq:cycleloss} and \eqref{eq:Disc}, respectively.
The resulting network architecture is shown in Fig.~\ref{fig:detailed_flowchart}(b).

\begin{figure*}[!hbt]
\center{ 
\includegraphics[width=18.0cm]{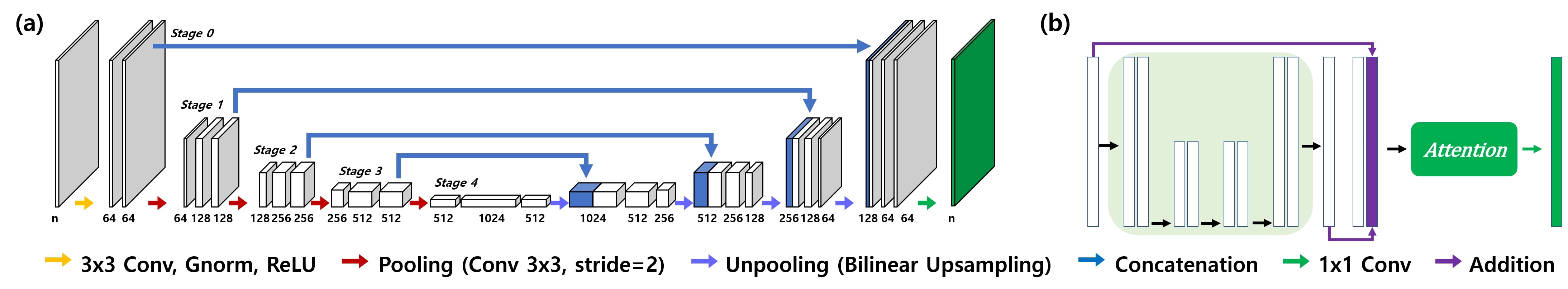}
}
\caption{Network architecture of the generators that were used in Step I and Step II. (a) Baseline U-Net which was modified from the original U-Net \citep{ronneberger2015u}. (b) Network architecture for $G_\Theta$ and $G_\Lambda$ which consist of baseline U-Net with adaptive residual attention module.}
\label{fig:generator_architecture}
\end{figure*}

Now, one of the main novelties in the second step comes from the design of the
 discriminator $\varphi$ in  \eqref{eq:Disc}.
 More specifically, to be an OT-cycleGAN,
 the discriminator $\varphi$ should satisfy the 1-Lipschitz condition, i.e.
 \begin{align}
|\varphi(\xb)-\varphi(\xb') | \leq \|\xb-\xb'\|,\quad \forall \xb,\xb'\in \Xc
 \end{align}
In this paper,   our discriminator architecture  is obtained from PatchGAN  as shown in Fig.~\ref{fig:discriminator_architecture}(b) \citep{zhu2017unpaired}.
However, care should be taken since $\Xc$ is composed of the 3D slabs.
Accordingly, slice direction is stacked in the channel dimension, so that 2-D convolution in PatchGAN can be utilized directly. In the first path, as shown in Fig.~\ref{fig:discriminator_architecture}(a), volume data is directly used as input to PatchGAN.
In the second path, max pooling is applied along the slice directions to generate the 2-D image, which is then used as an input
for PatchGAN (see Fig.~\ref{fig:discriminator_architecture}(a)).
This is in fact equivalent to applying the PatchGAN to the MIP image at each slab,
which is   necessary for learning the distribution of  MIP. The quality of MIP images are important in that MIP images are primarily used for radiologists in search of vascular pathology. Although equally important, source images usually serve as a supplementary tool.

Mathematically, the resulting discriminator $\varphi$ can be represented
as
\begin{align}
{\varphi(\xb)=\lambda_1\varphi_{1}(\xb)+\lambda_2\varphi_{2}^{\max}(\xb)}
\end{align} 
{where $\varphi_{1}$ and $\varphi_{2}^{\max}$
are discriminators for the original volume and max-pooled images, respectively, and $\lambda_1$, and $\lambda_2$ are appropriate hyperparameters.}
Then, the resulting discriminator loss function in \eqref{eq:Disc} can be decomposed as follows: 
\begin{eqnarray}
\begin{split}
&\ell_{Disc}(G, F;\varphi,\psi) \\
&= {\lambda_1\left(\int_\Xc \varphi_{1}(\xb)d\mu(\xb) - \int_\Yc \varphi_{1}(G_\Theta(\yb))d\nu(\yb)\right)} \\
&+ \lambda_2\left(\int_\Xc \varphi_{2}^{\max}(\xb)d\mu(\xb) - \int_\Yc \varphi_{2}^{\max}(G_\Theta(\yb))d\nu(\yb)\right) \\
&+ \left(\int_\Yc \psi(\yb)d\nu(\yb) - \int_\Xc \psi(F \xb)d\mu(\xb)\right)
\end{split}
\label{eq:WGAN loss_1}
\end{eqnarray}
Here, the generators $G$ and $F$ are implemented using
neural network parameterized by $\Lambda$ and $\Psi$, respectively,
whereas the discriminators $\varphi = \lambda_1\varphi_1+\lambda_2\varphi_2$ and $\psi$ are realized
using neural network with the weights $\Upsilon =[\Upsilon_1,\Upsilon_2]$ and $\Xi$, respectively.

 By jointly optimizing the set of  discriminators responsible for learning the distribution of the stacked volume, and the MIP discriminator which learns the distribution of MIP images, our method greatly improves the quality of MIP images whilst keeping the integrity of the source images.
See Fig.~\ref{fig:detailed_flowchart}(b) for the overall architecture of Step II reconstruction.

\section{Methods}
\label{sec:method}

\subsection{Training Dataset}

{From 10 patients who volunteered for scanning, 19 sets of in vivo data were acquired with 3T Philips Ingenia scanner.
Specifically, out of 10 patients, the scans were acquired as follows:
\begin{itemize}
    \item acceleration $\times$1 : 1 patient
    \item acceleration $\times$1, acceleration $\times$4 : 4 patients
    \item acceleration $\times$1, acceleration $\times$8 : 4 patients
    \item acceleration $\times$4, acceleration $\times$8 : 1 patient
\end{itemize}
In terms of number of slices used to train the neural network, a total of 18343 fully-acquired slices and 18356 under-sampled slices were used to train Step I neural network. 
For Step II training, 540 fully-acquired slices and 540 under-sampled slices were used to train the neural network.} 

All the scans were specified to the region covering the whole brain, with the field-of-view (FOV) of 180 x 180 mm. Specific parameters for the scans were defined as follows: repetition time (TR) = 23.00 ms, echo time (TE) = 3.45 ms, and FA = 18.00$^{\circ}$. Moreover, partial Fourier acquisition \citep{feinberg1986halving} was applied to the frequency encoding direction. Each set was acquired through MOTSA, consisting of 6 slabs, with  $k$-space matrix size 774x359x21 and 30 coils.
Once the k-space data are filled, the final reconstruction is obtained as 512x512x45 matrix size with zero padding and center cropping. {For training, 12 sets of patient data were used, while 7 sets of patient data were used for simulation study, and in vivo study.}

\begin{figure}[!hbt]
\center{ 
\includegraphics[width=6cm]{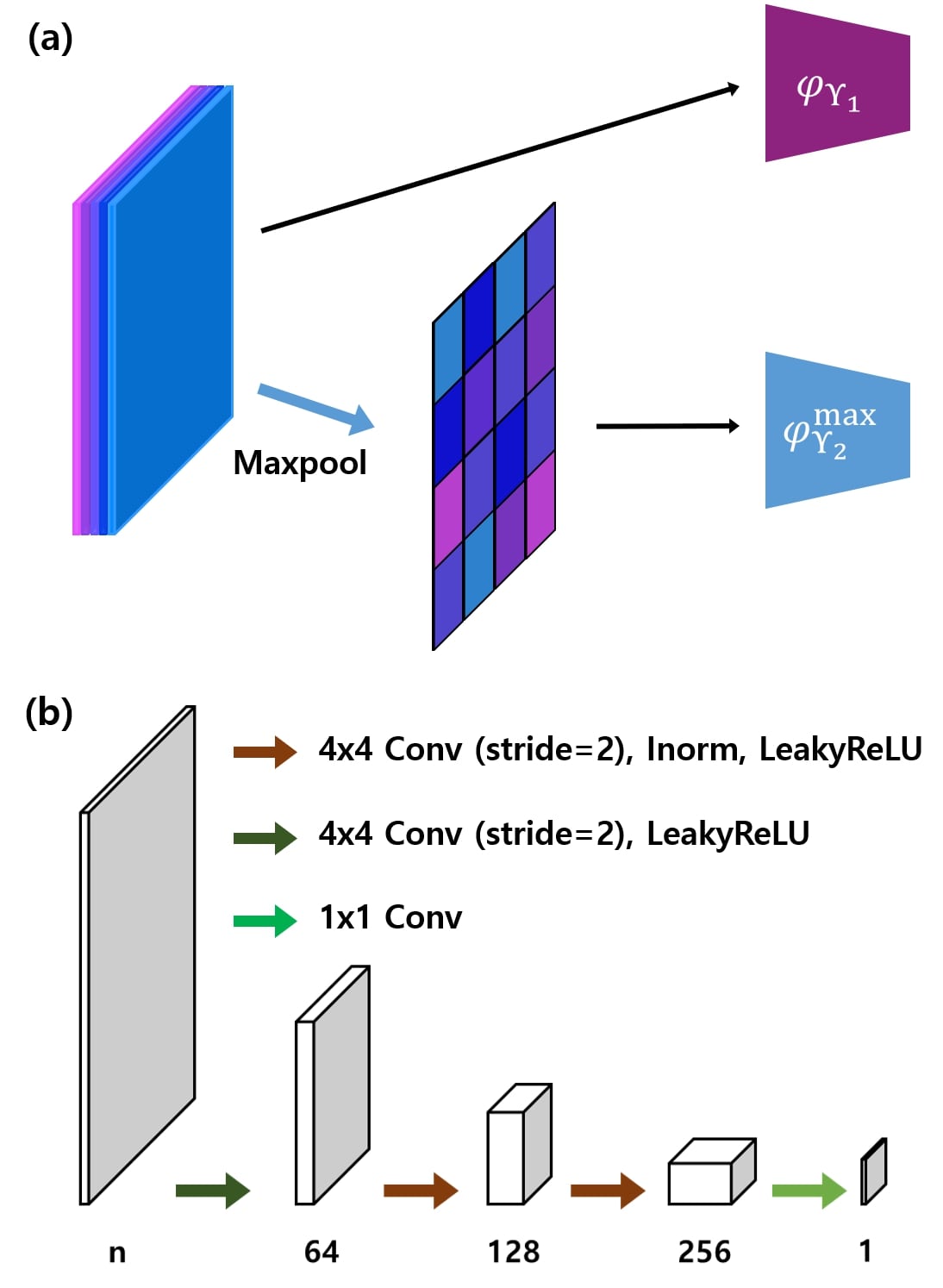}
}
\caption{Network architecture of the discriminators that were used in Step I and Step II. (a) Double-headed discriminator $\varphi_\Upsilon$ which consists of $\varphi_{\Upsilon_1}$ that takes in volume data, and $\varphi_{\Upsilon_2}^{max}$ that receives depthwise-maxpooled image as input. (b) Shared discriminator architecture that was adopted from patchGAN in \citep{zhu2017unpaired}. All the discriminators presented in our work share this specific architecture.}
\label{fig:discriminator_architecture}
\end{figure}

For the undersampling mask $\Omega$, the same masks that are used to accelerate MR scans from Philips Ingenia scanner were used without modification. Hence, two determined masks were used for x4 acceleration and x8 acceleration, respectively.

\subsection{Network Architecture}

\subsubsection{Generator Architecture}

For the single generator used in Step I training, we use modified U-Net architecture, which consists of four stages of convolutional layer, ReLU activation, and group normalization. Pooling and unpooling operations were constructed with 3x3 convolution with stride 2, and upscaling with bilinear interpolation, respectively. The number of convolutional filter channels was set to 64 at the first stage, and was increased two-fold at every stage, reaching 1024 at the last stage. To cope with the inherent nature of MR data which are complex, we stick to the conventional notion by stacking real and imaginary parts in the channel dimension. Thus, the dimension of the input channel was set to 60 (30 coils $\times$ 2 = 60). For detailed description, see Fig.~\ref{fig:generator_architecture}.

Moreover, we utilize nonlinear attention module which is known to enhance the expressivity of the network \citep{cha2020geometric}. For $G_\Theta$ in Step I, due to the large discrepancy between the input and the desired distribution, we utilize the same network architecture from Fig.~\ref{fig:generator_architecture}(a) as the attention module. Moreover, in $G_\Lambda$, a single 1$\times$1 convolution layer is utilized as the attention module.

In Step II training, we used two separate architectures for the mapping $G_\Lambda$ and $F_\Psi$. For the generator $G_\Lambda$, which is crucial, we adopt U-Net architecture as in Fig. \ref{fig:generator_architecture}(a) and set the initial filter length as 32 with 3 stages.
The network input is 3D volume composed of multiple slice images, which are stacked along the channel direction.
The network output is enhanced 3D volume with the same number of the slices. Slice depth of 7 was used, whose choice will be discussed further in the discussion section.
For the generator $F_\Psi$, we set the initial filter length to 8 with only 2 stages, restricting the expressivity of the network. Differentiating the two networks by the size resulted in more efficient and stable training compared to when we used two identical networks.
Again, the input and output of $F_\Psi$ is also three dimensional volume, where each slice is stacked along the channel direction.

\subsubsection{Discriminator Architecture}
  
The discriminators used in both steps were adopted from \citep{zhu2017unpaired}, and was modified to stabilize the training process. Specifically, patchGAN with 4x4 convolution kernel of three stages was used. Each stage consists of convolutional layer, instance normalization and leaky ReLU activation function as shown in Fig.\ref{fig:discriminator_architecture}(b). Moreover, spectral normalization \citep{miyato2018spectral} was applied to each layer for stability.

Discriminator architecture in Step II training is depicted in Fig. \ref{fig:discriminator_architecture}(a). For the given volume data, $\varphi_\Upsilon$ has two paths: $\varphi_{\Upsilon_1}$ which directly receives the volume as input, and $\varphi_{\Upsilon_2}^{\max}$ which collects single slice images acquired from maxpooling operation. $\varphi_{\Upsilon_1}$ and $\varphi_{\Upsilon_2}^{\max}$ can be seen as a double-headed discriminator $\varphi_\Upsilon$ as depicted in Fig. \ref{fig:discriminator_architecture} (b).

\subsection{Network Training}

For the first step of training, hyperparameters in \eqref{eq:phaseIloss} were set to $\gamma = 100$, $\alpha$ = 0.5, $\beta$ = 1. For optimization, RAdam optimizer \citep{liu2019variance}, \citep{kingma2014adam} was used with together with lookahead optimizer \citep{zhang2019lookahead}. Parameters for RAdam were set to $\beta_1$ = 0.5, $\beta_2$ = 0.999. Parameters for lookhead were set to $k$ = 5, $\alpha$ = 0.5. The initial learning rate was set to 0.0001 and was trained for 100 epochs. At 60 epoch of training, learning rate was decayed by a magnitude of 0.1. 

For Step II traning, hyperparameters in \eqref{eq:WGAN loss_1} were set to $\lambda_1$ = 5 and $\lambda_2$ = 3.
In the second step, Adam optimizer \citep{kingma2014adam} was used with parameters $\beta_1$ = 0.5 and $\beta_2$ = 0.999. 100 epochs of training was performed with consistent learning rate of 0.0001. 

For both steps of training, each input data was divided with the standard deviation of each input a priori. The proposed method was implemented in Python using PyTorch \citep{paszke2017automatic} with NVidia GeForce GTX 2080-Ti graphics processing unit. For the first step, the training took about three days, while the training of the second step took about 4 hours.

\section{Results}
\label{sec:results}

\subsection{Simulation study}

To verify the feasibility of our proposed method, and to prove that our method does not artificially generate pseudo-structures or pseudo-lesions that are not present in the ground truth, we first performed a reconstruction using  retrospectively subsampling. First, we retrospectively subsampled fully acquired k-space data with the given masks, each responsible for acceleration factor of $\times$4 and $\times$8.  
The undersampled k-space were subsequently reconstructed with the proposed method with trained $G_\Theta$ and $G_\Lambda$. {Here, Fig.~\ref{fig:main_results_in_vitro}(a) refers to the results achieved from two-step supervised learning. More specifically, the same neural network architectures used in the proposed method, $G_\Theta$ and $G_\Lambda$, were trained as a two-step process - in the coronal plane and the axial plane. Images in Fig.~\ref{fig:main_results_in_vitro}(b) column shows results with Step I of the proposed method, where only the reconstruction in the coronal plane was utilized. Fig.~\ref{fig:main_results_in_vitro}(c) contains results from our proposed method, where reconstruction took place both in coronal and axial directions.  

Moreover, when we compare results that were reconstructed with a two-step supervised learning process shown in Fig.~\ref{fig:main_results_in_vitro}(a), our proposed method shows superiority in preserving texture and realistic vessel structures. Results reconstructed with supervised learning tend to be blurry and the background near vessels contain more noise, whereas with the proposed method we can reconstruct high-resolution images with clear vessel structure.}
In fact, this kind of over-smoothing is quite often reported in supervised learning for image reconstruction.
On the other hand, unsupervised learning approaches without matched reference data should learn the distributions, so the oversmoothing
by fitting too much on the target data can be avoided.

\begin{figure*}
\center{ 
\includegraphics[width=18.0cm]{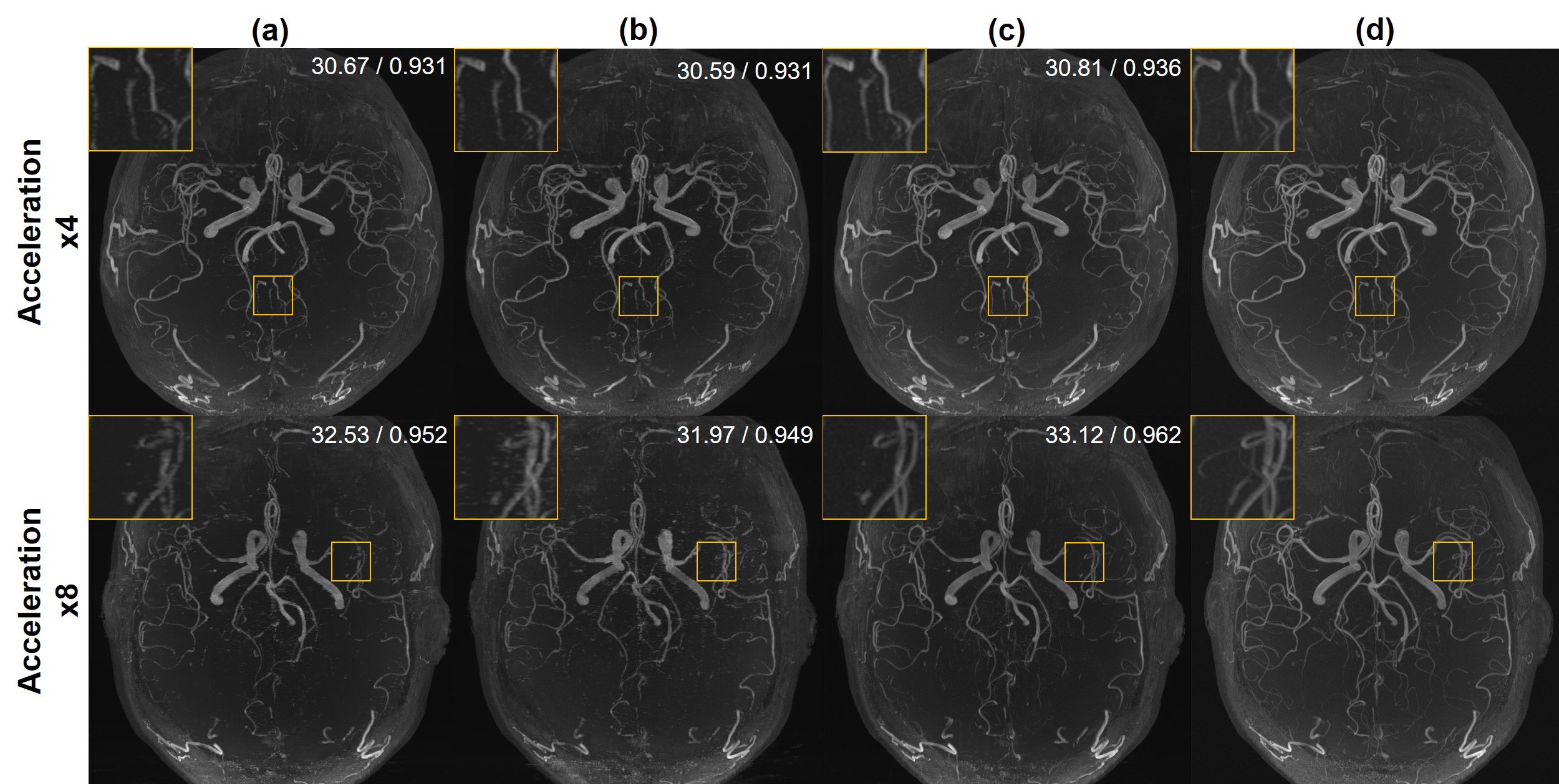}
}
\caption{MIP images from retrospective subsampling that were reconstructed using our method with single step and multi step. (a)
refers to the reconstruction using a  two step supervised learning approaches,  (b) refers to reconstructions where only Step I unsupervised learning was performed. (c) refers to reconstructed results after both Step I and II. (d) shows label images. The first row compares results from $\times$4 acceleration, while the second row compares results from $\times$8 acceleration. White numbers in the upper right part of the images indicate PSNR and SSIM, respectively.}
\label{fig:main_results_in_vitro}
\end{figure*}

\begin{figure*}
\center{ 
\includegraphics[width=15.0cm]{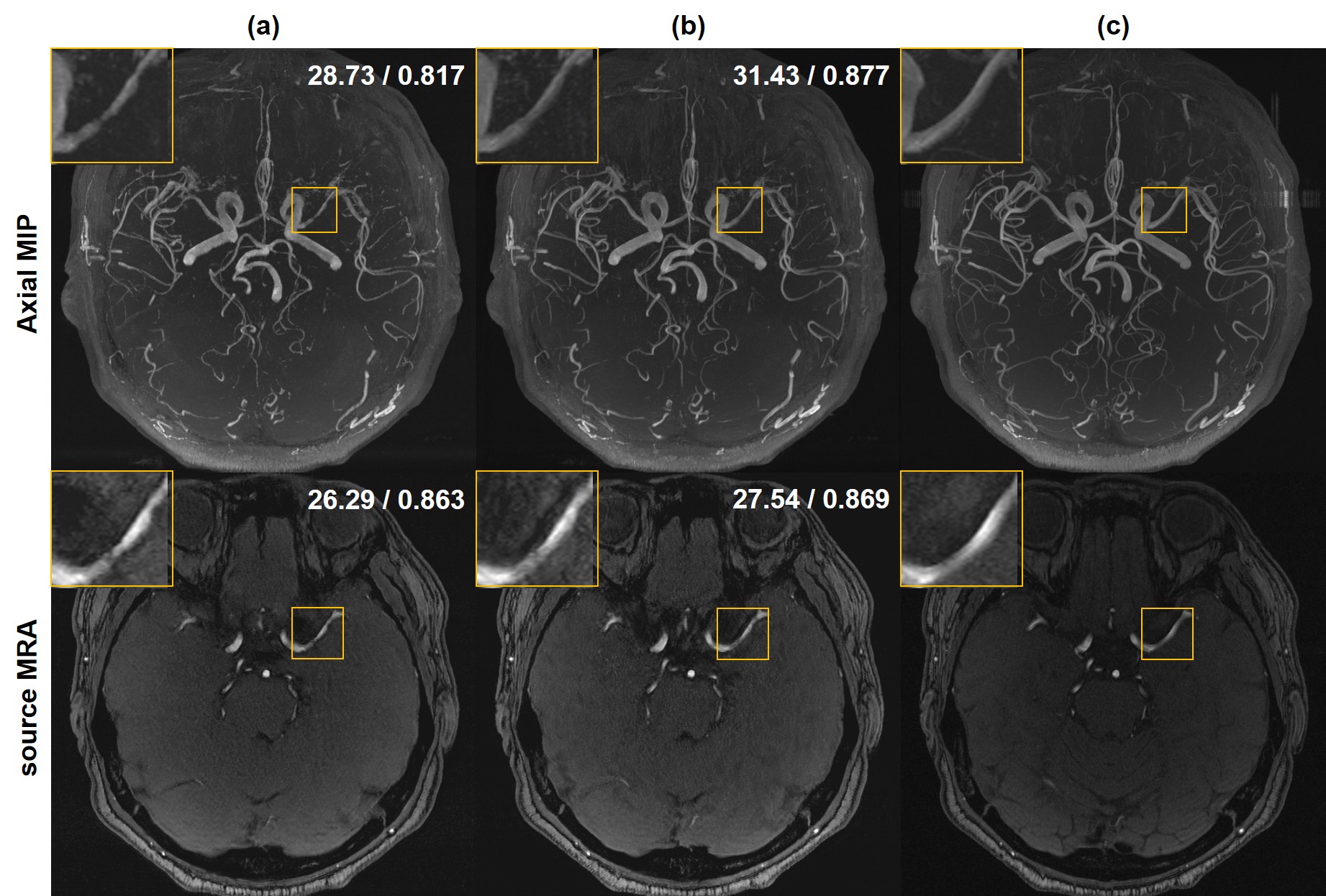}
}
\caption{Reconstruction results from retrospective subsampling  with and without $\varphi_{\Upsilon_2}^{\max}$. (a) indicates reconstructions that were performed in both steps, but without the projection discriminator. (b) shows results of our proposed method, with $\varphi_{\Upsilon_2}^{\max}$ present. (c) is the label data.  White numbers in the upper right part of the images indicate PSNR and SSIM, respectively. The yellow arrows in the figure indicate visible vessel structure with the proposed method, which was not visible with the reconstruction without the projection discriminator}
\label{fig:comparison_projection_D}
\end{figure*}

\begin{figure}
\hspace*{-0.3cm}
\centerline{ 
\includegraphics[width=9.3cm]{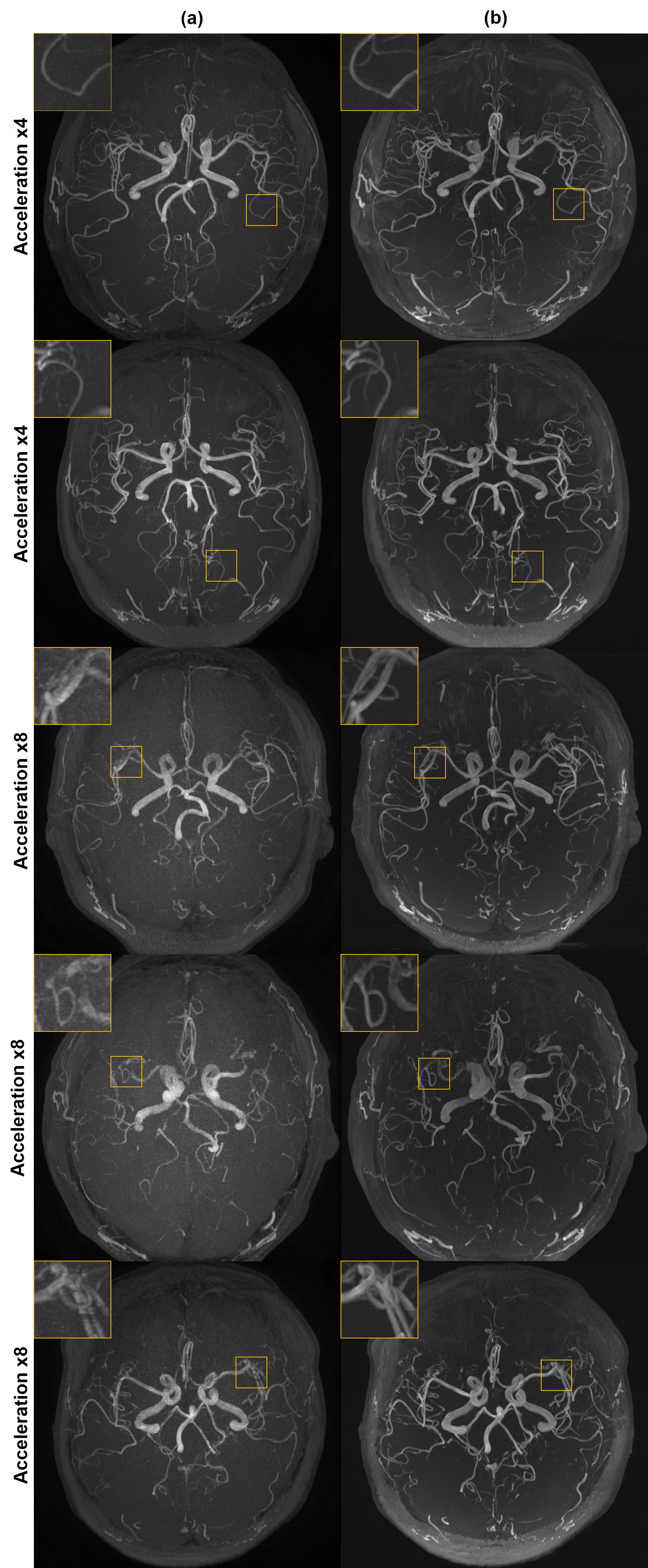}
}
\caption{In vivo reconstruction results viewed from the axial plane from acceleration factor of $\times$4 and $\times$8. MIP was also performed in the axial plane. (a) shows results from the compressed SENSE algorithm of the vendor. (b) shows results from our proposed method.}
\label{fig:main_results}
\end{figure}

 Furthermore, the MIP image reconstructed with single step training has numerous discontinuous vessels that are hard to distinguish from lesions, as shown in Fig.  \ref{fig:main_results_in_vitro}(b). 
In contrast, results from multiplanar reconstruction as shown in Fig. \ref{fig:main_results_in_vitro}(c), clearly have more visible vessels that are connected, and vascular discontinuity that was observed from uniplanar learning cannot be seen. Through two step learning, the vessel structures are much better preserved, not to mention the texture and detailed structures that closely resemble label images.
The advantage of the two step learning can best be seen in the MIP images.
From Fig. \ref{fig:main_results_in_vitro}, we also verify that artificial structures are not generated from our algorithm. Even though the acceleration factor in Fig. \ref{fig:main_results_in_vitro} is $\times$8, with the proposed method we are able to reconstruct images that faithfully resemble the structures shown in the label images.

To inspect the effect of $\varphi_{\Upsilon_2}^{\max}$, we also compare results without it. Although results without using $\varphi_{\Upsilon_2}^{\max}$ show improvement as opposed to results from uniplanar learning, they fall short behind our proposed method, especially in MIP image where we can still see pseudo-stenosis in the first row of Fig. \ref{fig:comparison_projection_D}. Visual clarity of vessels is also enhanced in source MRA images (second row, Fig. \ref{fig:comparison_projection_D}), where we see a thin vessel structure that is not apparent in the image shown in the third column.

\begin{figure*}[!ht] 	
\center{ 
\includegraphics[width=15.0cm]{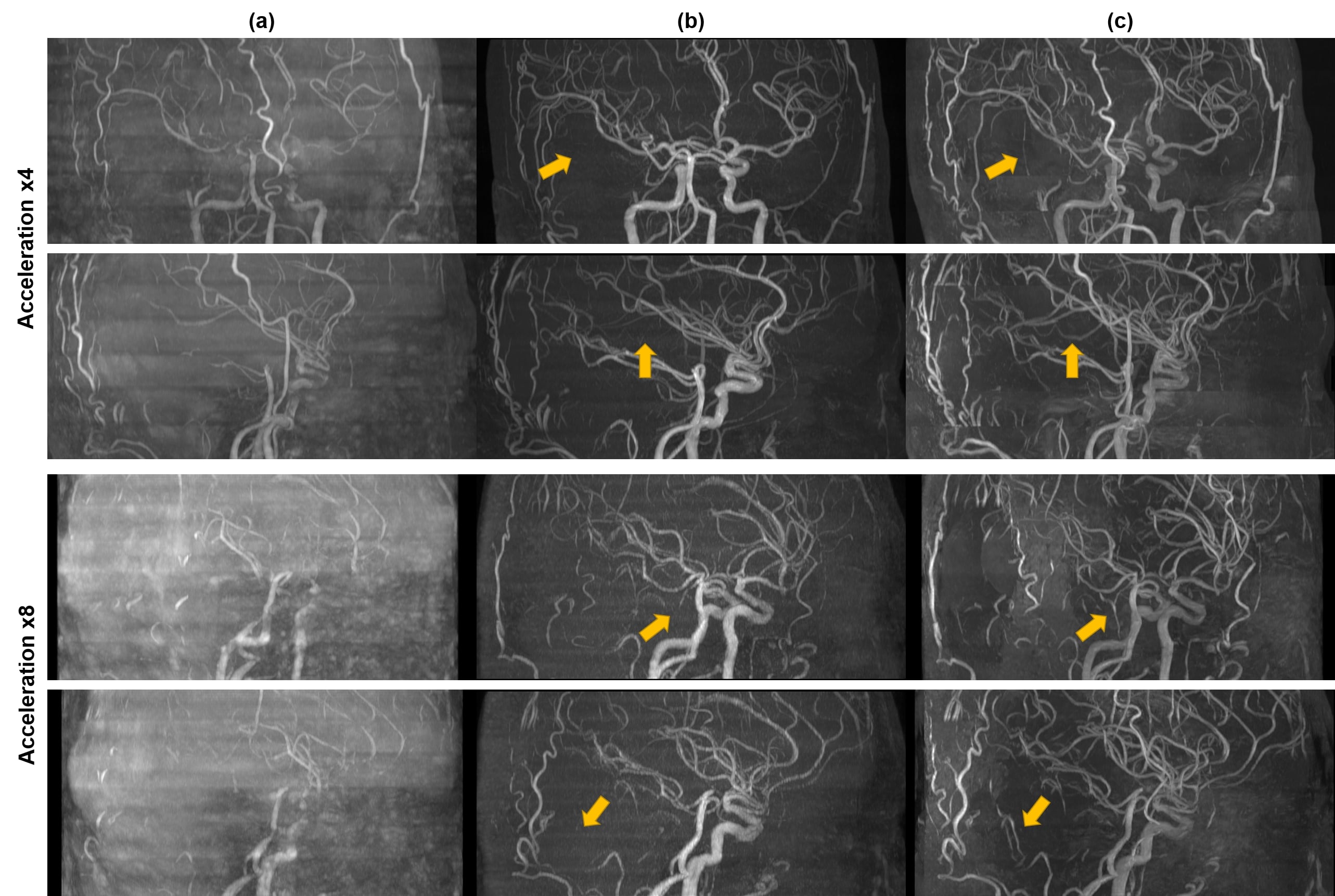}
}
\caption{In vivo MIP from multiple angles are presented in the figure. (a) refers to zero-filled reconstructions. (b) shows images that were directly acquired from the vendor, which are reconstructed using a CS algorithm (compressed SENSE). Images in (c) were reconstructed from raw k-space data using our proposed algorithm. The first two rows show reconstructed results from acceleration factor of $\times$4, while the latter two rows show results from acceleration factor of $\times$8.}
\label{fig:main_results_multiangle_MIP}
\end{figure*}

\subsection{In Vivo study}
To establish the improvements from the proposed method as opposed to conventional compressed sensing method that are used, we first performed an in vivo study where we compare reconstructions by the internal algorithm (Compressed SENSE \citep{geerts2018compressed}) that Philips 3T Ingenia scanner uses, to the reconstructions of ours.

As depicted in Fig. \ref{fig:main_results}, our method clearly demonstrates superiority with vessel contrast and continuity. Yellow arrows in Fig.~\ref{fig:main_results} show that our method is able to reconstruct what were not visible with the CS algorithm by the vendor. Moreover, robustness to noise, which impairs image quality, is also a clear advantage as opposed to the conventional algorithm. 

Furthermore, MIP from different angles as presented in Fig.  \ref{fig:main_results_multiangle_MIP} verifies that our proposed method clearly outperforms the algorithm of the vendor consistently in any projection directions. Namely, our method is able to reconstruct vessel structures that were not visible through the algorithm of the vendor, as marked with yellow arrows. Moreover, enhanced vascular continuity can be observed in the figure, which is important in clinical settings.

\subsection{Radiological evaluation}

The images were evaluated by a neuroradiologist (L.S.) with 10 years of experience in neuroimaging. The source images of TOF-MRA as well as the MIP images were assessed simultaneously during the evaluation.

On Fig.~\ref{fig:main_results_in_vitro}, when acceleration factor $\times$4 was applied, the MIPs of all four methods seem to be acceptable for relatively large blood vessels. However, when small vessels are evaluated (displayed as yellow boxes), the lumen of the vessels shows shaggy appearance on the image  with supervised learning approach (Fig.~\ref{fig:main_results_in_vitro}(a)) or with step I reconstruction only (Fig.~\ref{fig:main_results_in_vitro}(b)), as if in cases with severe atherosclerosis. In addition, very fine branches of vessels are missing on Fig.~\ref{fig:main_results_in_vitro}(a) and ~\ref{fig:main_results_in_vitro}(b), whereas it is faintly visualized on the proposed method (Fig.~\ref{fig:main_results_in_vitro}(c)), although slightly less conspicuous than on the label image (Fig.~\ref{fig:main_results_in_vitro}(d)).
When acceleration factor $\times$8 was applied, even the lumen of large vessels become irregular and discontinuous on Fig.~\ref{fig:main_results_in_vitro}(a) or ~\ref{fig:main_results_in_vitro}(b). Although fine branches of vessels are still missing on the proposed method (Fig.~\ref{fig:main_results_in_vitro}(c)), the lumen of the large vessels are well-visualized and acceptable for evaluation. The label image (Fig. ~\ref{fig:main_results_in_vitro}(d)) confirms that there are no pathology in the intracranial vessels.
On Fig.~\ref{fig:comparison_projection_D}, the image reconstructed with the projection discriminator (Fig.~\ref{fig:comparison_projection_D}(b)) depicts the contour of the vessel more clearly than the one without the projection discriminator (Fig.~\ref{fig:comparison_projection_D}(a)). Although, a focal mild pseud-stenosis is noted (arrowheads), the degree is much milder than on Fig.~\ref{fig:comparison_projection_D}(a) and can be easily dismissed considering the MRA source image. No stenosis is noted on the label image (Fig.~\ref{fig:comparison_projection_D}(c)).
When the images reconstructed by the proposed algorithm were compared to the ones by the vendor algorithm while maintaining the acceleration factor (Fig.~\ref{fig:main_results}), the proposed algorithm (Fig.~\ref{fig:main_results}(b)) was clearly superior to the vendor (Fig.~\ref{fig:main_results}(a)), in terms of signal-to-noise and conspicuity of the vessel contour. In particular, the vendor images with acceleration factor of $\times$8 seem to be unacceptable for clinical practice in its present form, where multiple pseudo-stenoses are found even for relatively large vessels.
Interestingly, the use of two-step reconstruction process, in the coronal plane followed by axial plane, appears to be helpful for reducing the so-called Venetian blind artifact, which is resultant to the MOTSA technique (Fig.~\ref{fig:main_results_multiangle_MIP}). The differences of signal intensity of adjacent slabs have been concomitantly adjusted to the image reconstruction.

\section{Discussion}
\label{sec:discussion}

\subsection{Optimal choice of slice depth}

With Step II training where we take partial stacks of volume data for training, we can flexibly choose the slice depth as a hyperparameter. To choose the optimal depth especially for constructing MIP, we experimented with slice depths $1 \sim 9$. The results in Table~\ref{tbl:ablation_slicedepth} using 3D data shows consistent improvement over using single slice data. Two reasons mainly account for this. First, the projection discriminator can no longer be utilized when we use depth 1 training. Since the main workhorse for improving the quality of MIP was the projection discriminator, the lack of this discriminator leads to poorer performance. Second, while reconstruction with depth 1 does improve the visual quality of the images by making the texture more realistic, it cannot enhance the visibility of vessels since information from adjacent slices are not accessible. Also, when we compare the metrics by varying the slice depth other than 1, we get the most effective result when we set the slice depth to 7. Table~\ref{tbl:ablation_slicedepth} indicates the choice of 7 as optimal slice depth is sound.

\begin{table}[!thb]
\caption{Comparison of quantitative metrics between reconstruction results of source MIP and MRA images with different slice depths. The number in each column indicates slice depths in the training of Step II.}
\centering
	\resizebox{0.45\textwidth}{!}{
\begin{tabular}{c|c|ccccc}
\hline
\multirow{2}{*}{\begin{tabular}[c]{@{}c@{}}Image\\ Type\end{tabular}} & \multirow{2}{*}{Metric} & \multicolumn{5}{c}{Number of slices} \\ \cline{3-7} 
                     &      & 1      & 3               & 5              & 7               & 9      \\ \hline\hline
\multirow{2}{*}{MIP} & PSNR & 29.27  & 31.02           & 29.92          & \textbf{31.43}  & 30.61  \\  
                     & SSIM & 0.8379 & \textbf{0.8774} & 0.8524         & 0.8771          & 0.8512 \\ \hline
\multirow{2}{*}{MRA} & PSNR & 29.23  & 29.42           & \textbf{31.10} & 30.00           & 29.11  \\ 
                     & SSIM & 0.7831 & 0.7723          & 0.7779         & \textbf{0.7958} & 0.7492 \\ \hline 
\end{tabular}
}
\label{tbl:ablation_slicedepth}
\end{table}

\subsection{Multiplanar learning vs. Volumetric learning}

There may be different ways to tackle 3D MR acceleration. Volumetric learning by utilizing full volume data could be a possible choice. Nonetheless, we propound that multiplanar learning is a better match for 3D TOF MRA reconstruction.

For one thing, GPU memory is limited, and loading the full 3D data into the GPU easily exceeds the constraint. Note that especially for multi-coil data where we have 4 dimensions in total: read-out, phase-encoding 1, phase-encoding 2, and coil, we have very limited size of data that are loadable to the GPU at once. In addition, with MOTSA scans where we have multiple slabs for each patient data, the choice of a {single volume} becomes ambiguous.

That being said, the proposed method that divides the training stage into two parts is a reasonable choice. Our method seamlessly incorporates all 4 dimensional information without technical overhead.

\section{Conclusion}
\label{sec:conclusion}

To devise a method that is well suited for the reconstruction of accelerated 3D TOF MRA, in this paper we suggested a multiplanar unpaired learning approach. In particular, MR-physics driven cycleGAN approach is exploited in the coronal plane as the first step of training process. Progressively, a novel cycleGAN approach in 3D with a newly-proposed projection discriminator is applied in the axial plane. The first step is meaningful in that we provide a method that is able to incorporate accelerated data into the training scheme, and by exploiting MR-physics we devise a method that is much stabler than the conventional cycleGAN approach. The second phase enhances the quality of images, especially images of MIP, which is more clinically meaningful. Our method can provide high quality reconstructions at very high acceleration factors which were not possible with conventional vendor CS methods. Thus, we suggest a new direction of study for the acceleration of 3D MRA by exploiting information from multiple axes without the need for large amount of paired data.

In this work, we used 7 patient data scans to validate the research. However, the number of scans used to test the proposed method is limited, and the method was not tested using scans in which lesions are apparent. Hence, to prove its clinical utility, a more comprehensive research in the clinical perspective using more data with enhanced diversity could be a further direction of research.

\section*{Acknowledgments}
This work was supported in part by Korea Advanced Institute of Science and Technology, Grant number N11200110, and in part by a grant from the National Research Foundation of Korea (NRF-2018R1C1B6007917 and NRF-2020R1A2B5B03001980) and by grants from the SNUBH Research Fund (No. 09-2019-006 and 16-2020-002).

\section*{Appendix}

The proof is a direct extension of the proof in \citep{sim2019OT}, but we include the following for self-containment.
%

Using the transportation cost $ c(\xb,\yb;\Theta)$ given by Eqs. \eqref{eq:cmain},
the primal optimal transport problem becomes
\begin{align}
\Kd(G,F):=&\min_{\pi\in \Pi(\mu,\nu)} \int_{\Xc\times \Yc} c(\xb,\yb;\Theta)d\pi(\xb, \yb) \\
=& \int_{\Xc\times \Yc}  c_{XY}(\xb,\yb) d\pi^*(\xb, \yb)  + \ell_x(G,F)+ \ell_y(G,F)
\end{align}
where $\pi^*$  denote the optimal joint measure, $\mu,\nu$ are the marginal distribution, and 
\begin{align*}
 c_{XY}(\xb,\yb) = &  \|\Ac(\xb)-\Ac(G(\yb))\|+ \|\Ac(F(\xb))-\Ac(\yb)\|  
\end{align*}
and 
\begin{align*}
 \ell_x(G,F) &= \min_{\pi\in \Pi(\mu,\nu)}\int_{\Xc\times \Yc} b_x\left(\xb;G,F\right)  d\pi(\xb, \yb) \\\
 &=\int_\Xc b_x\left(\xb;G,F\right)  d\mu(\xb)
\end{align*}
after integrating out with respect to $\yb$; similarly, we have
\begin{align*}
 \ell_y(G,F)
 &=\int_\Yc b_y\left(\yb;G,F\right)  d\nu(\yb)
\end{align*}
Now,  according to the Kantorovich dual formulation \citep{villani2008optimal}, we have
\begin{align*}
\Kd_{XY}:= & \int c_{XY}(\xb,\yb) d\pi^*(\xb, \yb) \\
= & \frac{1}{2}\left\{ \max_{\zeta}\int_\Xc\zeta(\xb)d\mu(\xb)+ \int_\Yc \zeta^c(\yb) d\nu(\yb)\right. \\
&+ \left.  \max_{\eta}\int_\Xc\eta^c(\xb)d\mu(\xb)+ \int_\Yc \eta(\yb) d\nu(\yb) \right\}
\end{align*}
where the so-called c-transforms $ \zeta^c(\yb)$ and $\eta^c(\xb)$ are defined by \citep{villani2008optimal}
\begin{align*}
 \zeta^c(\yb)
&= \inf_\xb \{  \|\Ac(\xb)-\Ac(G(\yb))\|+ \|\Ac(F(\xb))-\Ac(\yb)\|  -\zeta(\xb) \}\\
 \eta^c(\xb)
&= \inf_\yb \{  \|\Ac(\xb)-\Ac(G(\yb))\|+ \|\Ac(F(\xb))-\Ac(\yb)\|  -\eta(\yb) \}\
\end{align*}
Now, instead of finding the $\inf_\xb$, we choose $\xb=G(\yb)$.
Similarly, instead of finding the $\inf_\yb$, we choose $\yb=F \xb$. 
This leads to an upper bound: 
\begin{eqnarray*}
\Kd_{XY}&\leq & 
\frac{1}{2}\left(\ell_{cycle}(G,F)+ \ell_{Disc}(G,F; \zeta, \eta)\right)
\end{eqnarray*}
where
\begin{align}
\ell_{cycle}(G,F) =&  \int_\Xc  \|\Ac(\xb)-\Ac(G(F(\xb)))\| d\mu(\xb)  \notag\\
& +\int_\Yc  \|\Ac(F(G(\yb)))-\Ac(\yb)\|d\nu(\yb) \label{eq:cycle}\\
 \ell_{Disc}(G,F; \zeta, \eta) = & \max_{\zeta}\int_\Xc \zeta(\xb)  d\mu(\xb) - \int_\Yc \zeta(G(\yb))d\nu(\yb)  \notag\\
 &+ \max_{\eta}\int_{\Yc} \eta(\yb)  d\nu(\yb) - \int_\Xc \eta(F \xb)  d\mu(\xb)
\end{align}
Now, if we define 
\begin{align}\label{eq:equ}
\zeta(\xb):= \varphi(\Ac(\xb)),&\quad \eta(\yb):= \psi(\Ac(\yb))
\end{align}
for some 1-Lipschitz function $\varphi$ and $\psi$,
we have
\begin{align*}
 \zeta(\xb)-\zeta(G(\yb)) &= \varphi(\Ac(\xb))-\varphi(\Ac(G(\yb)))\\
 &\leq  \|\Ac(\xb)-\Ac(G(\yb))\| \\
 &\leq \|\Ac(\xb)-\Ac(G(\yb))\|+ \|\Ac(F(\xb))-\Ac(\yb)\| \\
 \eta(\yb)- \eta(F \xb)&= \psi(\Ac(\yb))-\psi(\Ac(F\xb))\\
&\leq \|\Ac(\xb)-\Ac(G(\yb))\|+ \|\Ac(F(\xb))-\Ac(\yb)\| \\
\end{align*}
This leads to the following lower-bound
\begin{eqnarray*}
\Kd_{XY}
&\geq& \frac{1}{2}\ell_{Disc}(G,F;\zeta,\eta)
\end{eqnarray*}
If we replace the discriminator using  \eqref{eq:equ},
we have
\begin{align*}
& \ell_{Disc}(G,F;\zeta,\eta) = \ell_{Disc}(G,F;\varphi,\psi) \\
:=&\max_{\varphi}\int_\Xc \varphi(\Ac(\xb))  d\mu(\xb) - \int_\Yc \varphi(\Ac(G(\yb)))d\nu(\yb) \notag \\
 & + \max_{\psi}\int_{\Yc} \psi(\Ac(\yb))  d\nu(\yb) - \int_\Xc \psi(\Ac(F(\xb)))  d\mu(\xb) \notag
\end{align*}
The rest of the proof is exactly the same as in \cite{sim2019OT}.
This concludes the proof.

\bibliographystyle{model2-names.bst}\biboptions{authoryear}

\end{document}